\documentclass[11pt]{article}
\usepackage[ascii]{inputenc}
\usepackage{amsmath}
\usepackage{amsfonts}
\usepackage{amssymb}
\usepackage{amsthm}
\usepackage{amscd}
\usepackage{fullpage}
\usepackage{ifthen,graphics,epsfig,url}

\theoremstyle{remark}
\newtheorem{lemma}{Lemma}[section]
\newtheorem{corollary}[lemma]{Corollary}
\newtheorem{theorem}[lemma]{Theorem}

\newtheorem{proposition}[lemma]{Proposition}

\DeclareMathOperator{\sk}{skel}
\DeclareMathOperator{\rk}{rank}
\DeclareMathOperator{\inv}{inv}
\DeclareMathOperator{\bd}{bdry}

\DeclareMathOperator{\ids}{ids}
\DeclareMathOperator{\vals}{vals}
\DeclareMathOperator{\id}{id}

\newcommand{\var}[1]{\lstinline+#1+}
\newcommand{\ang}[1]{\langle{#1}\rangle}

\newcommand{\set}[1]{\left\{ #1 \right\}}

\newcommand{\cA}{\ensuremath{\mathcal{A}}}

\newcommand{\cC}{\ensuremath{\mathcal{C}}}

\newcommand{\cG}{\ensuremath{\mathcal{G}}}

\newcommand{\cI}{\ensuremath{\mathcal{I}}}

\newcommand\cK{\ensuremath{\mathcal{K}}}
\newcommand{\cL}{\ensuremath{\mathcal{L}}}

\newcommand{\cO}{\ensuremath{\mathcal{O}}}
\newcommand{\cP}{\ensuremath{\mathcal{P}}}

\newcommand{\cS}{\ensuremath{\mathcal{S}}}

\title{An Equivariance Theorem with Applications to Renaming\\
(Preliminary Version)}
\author{
Armando Casta\~neda%
	\thanks{Instituto de Matem\'aticas,
		Universidad Nacional Aut\'onoma de M\'exico,
		Ciudad Universitaria, D.F.~04510,
		Mexico;
		acastanedar@uxmcc2.iimas.unam.mx.}
	\and
Maurice Herlihy%
	\thanks{Brown University,
		Computer Science Department,
		Providence, RI~02912;
		mph@cs.brown.edu. Supported by NSF 000830491.}
	\and
Sergio Rajsbaum%
	\thanks{Instituto de Matem\'aticas,
		Universidad Nacional Aut\'onoma de M\'exico,
		Ciudad Universitaria, D.F.~04510,
		Mexico;
		rajsbaum@math.unam.mx. Supported by UNAM-PAPIIT}
}
\begin{document}

\begin{titlepage}
\maketitle

\begin{abstract}
In the renaming problem,
each process in a distributed system is issued a unique name from a large name space,
and the processes must coordinate with one another to choose unique names from a much smaller name space.

We show that lower bounds on the solvability of renaming in an asynchronous
distributed system can be formulated as a purely topological question about
the existence of an equivariant chain map from a ``topological disk'' to a
``topological annulus''.
Proving the non-existence of such a map implies the non-existence of a
distributed renaming algorithm in several related models of computation.
\end{abstract}
~\\ 

%\noindent Contact person: Sergio Rajsbaum, \texttt{rajsbaum@math.unam.mx}.\\

%\noindent Regular submission.\\
\thispagestyle{empty} 
\end{titlepage}

\section{Introduction}
In the \emph{$M$-renaming} task,
each of $n+1$ processes is issued a unique name taken from a large namespace,
and after coordinating with one another,
each chooses a unique name taken from a (much smaller) namespace of size $M$.
Processes are \emph{asynchronous}
(there is no bound on their relative speeds),
and potentially \emph{faulty}
(any proper subset may halt without warning).
Assuming processes communicate through a shared read-write memory,
for which values of $M$ can we devise a protocol that ensures that all
non-faulty processes choose unique names?

To rule out trivial solutions,
we require that any such protocol be \emph{anonymous}:
informally stated,
in any execution,
the name a process chooses can depend only on the name it was originally issued
and how its protocol steps are interleaved with the others.

This problem was first proposed by Attiya et al.~\cite{abdprjacm90},
who provided a protocol for $M = 2n+1$,
and showed that there is no protocol for $M=n+2$.
Later, Herlihy and Shavit~\cite{hsjacm99} used chain complexes,
a construct borrowed from Algebraic Topology,
to show impossibility for $M = 2n$.
Unfortunately, this proof,
and its later refinements~\cite{arsiam02, hsjacm99, hrmscs00},
had a flaw:
because of a calculation error,
the proof did not apply to certain ``exceptional'' dimensions satisfying a
number-theoretic property described below.
Casta\~neda and Rajsbaum~\cite{crpodc08} provided a new proof based on
combinatorial properties of black-and-white simplicial colorings,
and were able to show that in these ``exceptional" dimensions,
and only for them, protocols do exist for $M = 2n-1$.
Nevertheless, this later proof was highly 
specialized for the weak symmetry breaking task,
a task equivalent to $2n$-renaming, 
so it was difficult to compare it directly to earlier proofs,
either for renaming, or for other distributed problems.
In the \emph{weak symmetry breaking} task~\cite{grhdisc06, hsjacm99}, 
each of $n+1$ processes chooses a binary output value, $0$ or $1$, such that
there is no execution in which the $n+1$ processes choose the same value.

The contribution of this paper is to formulate the complete renaming proof
entirely in the language of Algebraic Topology,
using chain complexes and chain maps.
While this proof requires more mathematical machinery than the specialized
combinatorial arguments used by Casta\~neda and Rajsbaum, 
the chain complex formalism is significantly more general.
While earlier work has focused on protocols for an asynchronous model where but
one process may fail (``wait-free'' protocols),
the chain complex formalism applies to any model where one can compute the
connectivity of the ``protocol complexes'' associated with that model.
This approach has also proved broadly applicable to a range
of other problems in Distributed Computing~\cite{hrtpodc98,hrmscs00}.
In this way,
we incorporate the renaming task in a broader framework of distributed problems.

As in earlier work~\cite{hrtpodc98,hrmscs00},
the existence (or not) of a protocol is equivalent to the existence of a certain
kind of chain map between certain chain complexes.
Here, we replace the \emph{ad-hoc} conditions used by prior work~\cite{hsjacm99, hrmscs00}
to capture the informal notion of anonymity with
the well-established mathematical notion of \emph{equivariance}.
Roughly speaking, a map is equivariant if it commutes with actions of a group
(in this case, the symmetric group on the set of process IDs).
We prove a purely topological theorem characterizing when there exists an
equivariant map between the chain complexes of an $n$-simplex and the chain
complexes of an annulus.
The desired map exists in dimension $n$ if and only if the binomial coefficients
${n+1 \choose 1}, \hdots, {n+1 \choose n}$
are relatively prime.
These are exactly the dimensions for which renaming is possible for $M = 2n$.

\section{Distributed Computing}
\label{sec:distributed}
We consider a distributed system of $n+1$ processes
with distinct IDs taken from $[n] = \{ 0, \hdots, n \}$.
Processes are \emph{asynchronous}: there is no restriction on their relative speeds.
They communicate by writing and reading a shared memory. 
A \emph{task} is a distributed problem where each process is issued a private \emph{input value},
communicates with the other processes,
and after taking a bounded number of steps,
chooses a private \emph{output value} and halts.
A \emph{protocol} is a distributed program that solves a task.
A protocol is \emph{$t$-resilient} if it tolerates crash failures by $t$ of fewer processes,
and it is \emph{wait-free} if it tolerates crash failures by $n$ out of $n+1$ processes.

We model tasks and distributed systems using notions from combinatorial
topology~\cite{arsiam02, hsjacm99}. 
An initial or final state of a process is modeled as a \emph{vertex},
a pair consisting of a process ID and a value (either input or output).
We speak of the vertex as \emph{colored} with the process ID.
A set of $d + 1$ mutually compatible initial or final states is modeled as a
\emph{$d$-dimensional simplex}, or \emph{$d$-simplex}.
It is \emph{properly colored} if the process IDs are distinct.
A nonempty subset of a simplex is called a \emph{face}. 
An $n$-simplex has ${n+1 \choose i+1}$ faces of dimension $i$.

The complete set of possible initial (or final) states of a distributed task
is represented by a set of simplexes, closed under containment,
called a \emph{simplicial complex}, or \emph{complex}. 
The \emph{dimension} of a complex $\cK$ is the dimension of a simplex of largest dimension in $\cK$. 
We sometimes use superscripts to indicate dimensions of simplexes and complexes. 
The set of process IDs associated with a simplex $\sigma^n$ is denoted by $\ids(\sigma^n)$, 
and the set of values by $\vals(\sigma^n)$.
Sometimes we abuse notation by using $\sigma$ to stand for the complex
consisting of $\sigma$ and its faces.
The \emph{boundary complex} $\bd \sigma$ is the complex consisting of proper
faces of $\sigma$.
For a complex $\cK$, its \emph{$i$-skeleton}, denoted $\sk^i(\cA)$, is the complex
containing all simplexes of $\cA$ of dimension at most $i$.

Any simplicial complex has a geometric realization as a point set in a Euclidean space.
A vertex corresponds to a point, and a simplex to the convex hull of affinely-independent vertexes.
A complex corresponds to the union of its geometric simplexes,
where any two geometric simplexes intersect either in a common face, or not at all.

A \emph{task} for $n+1$ processes consists of an \emph{input complex} $\cI^n$,
and output complex $\cO^n$ and a map $\triangle$ carrying each input $n$-simplex of $\cI^n$
to a set of $n$-simplexes of $\cO^n$.
This map associates with each initial state of the system (an input $n$-simplex)
the set of legal final states (output $n$-simplexes).
It is convenient to extend $\triangle$ to simplexes of lower dimension: 
$$\triangle(\sigma^m) = \cap \triangle(\sigma^n)$$
where $\sigma^n$ ranges over all $n$-simplexes containing $\sigma^m$. This definition has the following
operational interpretation: $\triangle(\sigma^m)$ is the set of legal final states in executions where only
$m+1$ out of $n+1$ processes participate (the rest fail without taking any steps). A protocol
\emph{solves} a task if when the processes run their programs, they start with mutually compatible
input values, represented by a simplex $\sigma^n$, communicate with one another, and eventually
halt with some set of mutually compatible output values, representing a simplex in $\triangle(\sigma^n)$.

Any protocol has an associated \emph{protocol complex} $\cP$,
in which each vertex is labeled with a process id and that process's
final state (called its \emph{view}).
Each simplex thus corresponds to an equivalence class of executions that
``look the same'' to the processes at its vertexes.
The protocol complex corresponding to executions starting from an input simplex
$\sigma^m$ is denoted ${\cP}(\sigma^m)$.

A \emph{vertex map} carries vertexes of one complex to vertexes of another.
A \emph{simplicial map} is a vertex map that preserves simplexes.
A simplicial map on properly colored complexes is \emph{color-preserving} if it associates vertexes of the same color.
Let $\cP$ be the protocol complex for a protocol. 
A protocol \emph{solves} a task $\langle \cI^n, \cO^n, \triangle \rangle$ if and only if there exists
a color-preserving simplicial map $\delta : \cP \to \cO^n$, 
called a \emph{decision map}, such that for every $\sigma^m \in \cI^n$, 
$\delta(\cP(\sigma^m)) \subset \triangle(\sigma^m)$.
We prove our impossibility results by exploiting the topological properties of the protocol complex and
the output complex to show that no such map exists.

\section{Algebraic Topology}
Here is a review of some basic notions of algebraic topology
(see Munkres~\cite{Munkres}, Hatcher~\cite{Hatcher} or Dieck~\cite{Dieck}).

Let $\sigma = \{ v_0, v_1, \hdots, v_q \}$ be a simplex.
An \emph{orientation} of $\sigma$ is a set consisting of a sequence 
of its vertexes and all even permutations of this sequence. 
If $n > 0$ then these sets fall into two equivalence classes, the sequence 
$\langle v_0 v_1 \hdots v_n \rangle$ and its even permutations,
and $\langle v_1 v_0 \hdots v_n \rangle$ and its even permutations.
Simplexes are oriented in increasing subscript order unless stated otherwise.

A \emph{$q$-chain} for a complex $\cK$ is a formal sum of oriented $q$-simplexes:
$\sum_{j=0} \lambda_j \sigma^q_j$, where $\lambda_j$ is an integer.
Simplexes with zero coefficients are usually omitted,
unless they are all zero, when the chain is denoted $0$. 
We write $1 \cdot \sigma^q$ as $\sigma^q$ and
 $-1 \cdot \sigma^q$ as $-\sigma^q$. 
For $q > 1$, $-\sigma^q$ is identified with $\sigma^q$ having the opposite orientation.
The $q$-chains of $\cK$ form a free Abelian group under component-wise addition,
called the \emph{$q$-th chain group} of $\cK$, denoted $\cC_q(\cK)$.
For dimension $-1$, we adjoin the infinite cyclic group $\mathbb{Z}$, $\cC_{-1}(\cK) = \mathbb{Z}$.

%\paragraph{Boundary Operators.}
A \emph{boundary} operator $\partial_q : \cC_q(\cK) \to \cC_{q-1}(\cK)$ 
is a homomorphisms that satisfies 
$$\partial_{q-1} \partial_q \alpha = 0$$
and an \emph{augmentation} $\partial_0 : \cC_0(\cK) \to \cC_{-1}(\cK)$ 
which is an epimorphism.
For an oriented simplex $\sigma = \{ v_0, v_1, \hdots, v_q \}$, let $face_j(\sigma)$ be the $(q-1)$-face 
of $\sigma$ without vertex $v_j$, i.e.,  $face_j(\sigma) = \{ v_0, \hdots, \hat{v_j}, \hdots, v_q \}$,
where circumflex ($\, \, \widehat{ } \, \,$) denotes omission.
For $q > 0$, the usual boundary operator 
$\partial_q : \cC_q(\cK) \to \cC_{q-1}(\cK)$ is defined on simplexes: 
$$\partial_q \sigma = \sum^q_{j=0} (-1)^j face_j(\sigma)$$
Boundary $\partial_q$ extends additively to chains: $\partial_q(\alpha + \beta) = \partial_q \alpha + \partial_q \beta$.
For $q=0$, $\partial_0(v) = 1$ and extend linearly. We sometimes omit subscripts from boundary operators.

A $q$-chain $\alpha$ is a \emph{boundary} if $\alpha = \partial \beta$ for some $(q+1)$-chain $\beta$,
and it is a \emph{cycle} if $\partial \alpha = 0$. Since $\partial \partial \alpha = 0$, every boundary is a cycle.

%\paragraph{Chain Complexes and Chain Maps.}
The \emph{chain complex} $\cC(\cK)$ of $\cK$, is the sequence of groups and homomorphisms 
\begin{equation*}
\set{ \cC_q(\cK), \partial_q }.
\end{equation*}
Let $\{ \cC_q(\cK), \partial_q \}$ and $\{ \cC_q(\cL), \partial'_q \}$
 be chain complexes for $\cK$ and $\cL$.
A \emph{chain map} $\phi$ is a family of homomorphisms 
$\phi_q: \cC_q(\cK) \to \cC_q(\cL)$,
that satisfies $\partial'_q \circ \phi_q = \phi_{q-1} \circ \partial_q$.
Therefore, $\phi_q$ preserves cycles and boundaries. 
That is, if $\alpha$ is a $q$-cycle ($q$-boundary) of $\cC_q(\cK)$,
$\phi_q(\alpha)$ is a $q$-cycle ($q$-boundary) of $\cC_q(\cL)$.
Any simplicial map $\mu: \cK \to \cL$ induces a chain map 
$\mu_\#: \cC(\cK) \to \cC(\cL)$.
(For brevity, $\mu$ denotes both the simplicial map and $\mu_\#$.)
Similarly, any subdivision induces a chain map.

Let $\cK$ and $\cL$ be properly-colored complexes.
A chain map $\phi: \cC(\cK) \to \cC(\cL)$
is \emph{color-preserving} if each $\tau \in a(\sigma)$ is properly colored
with the colors of $\sigma$.

Let $\cG$ be a finite group and $\cC(\cK)$ be a chain complex.
An \emph{action} of $\cG$ on $\cC(\cK)$ is a set 
$\Phi = \{ \phi_g \vert g \in \cG \}$ of chain maps 
$\phi_g: \cC(\cK) \rightarrow \cC(\cK)$ such that
for the unit element $e \in \cG$, 
$\phi_e$ is the identity, and for all $g, h \in \cG$,
$\phi_g \circ \phi_h = \phi_{gh}$. For clarity, we write 
$g(\sigma)$ instead of $\psi_g(\sigma)$.
The pair $(\cC(\cK), \Phi)$ is a \emph{$\cG$-chain complex}.
When $\Phi$ is understood, we just say that $\cC(\cK)$ is a
$\cG$-chain complex. 

Consider two $\cG$-chain complexes 
$(\cC(\cK), \Phi)$ and $(\cC(\cL), \Psi)$.
Suppose we have a family of homomorphisms 
$$\mu_q: \cC_q(\cK) \rightarrow \cC_p(\cL)$$
possibly $q \neq p$. We say that $\mu = \{ \mu_q \}$ is \emph{$G$-equivariant},
or just \emph{equivariant} when $\cG$ is understood, if 
$\mu \circ \phi_g = \psi_g \circ \mu$
for every $g \in G$.
This definition can be extended to a family of homomorphisms as follows.
For each dimension  each $q$
suppose we have a family of homomorphisms
$$\mu^1_q, \hdots, \mu^{i_q}_q: \cC_q(\cK) \to \cC_p(\cL)$$
We say that $\mu = \{ \mu^{i_q}_q \}$ is \emph{$G$-equivariant} if 
for every $g \in G$ and for every $\mu^i \in \mu$,
$\mu^j \circ \phi_g = \psi_g \circ \mu^i$ for some $\mu^j \in \mu$.

Let $\cS_n$ be the symmetric group consisting of all permutations of $[n]$.
Henceforth, ``equivariant'' means ``$\cS_n$-equivariant'',
where the value of $n$ should be clear from context.

\section{Weak Symmetry-Breaking}
It is convenient to reduce the $2n$-renaming problem to the following
equivalent~\cite{grhdisc06} but simplified form.
In the \emph{weak symmetry-breaking} (WSB)  task~\cite{grhdisc06, hsjacm99},
the processes start with trivial inputs,
and must choose $0$ or $1$ such that  not all decide $0$ and not all decide $1$.
Just as for renaming,
to rule out trivial solutions any protocol for WSB must be anonymous.

\begin{figure}[ht]
    \begin{center}
        \includegraphics[scale=.65]{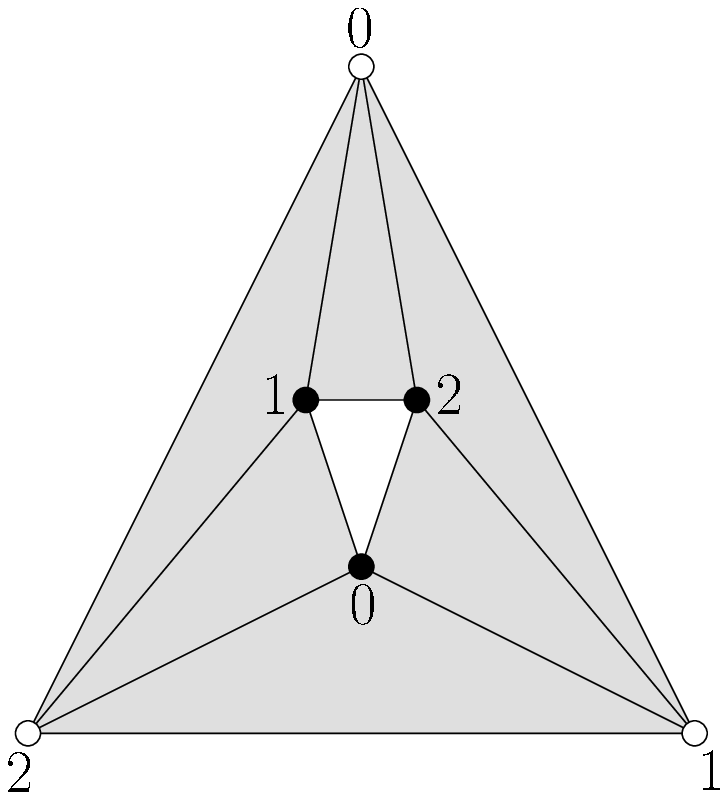}
        \caption{The annulus of dimension $2$.}
		\label{fig:annulus}
    \end{center}
\end{figure}

We are interested in two complexes: the input and output complexes for weak symmetry-breaking.
Topologically, the input complex a combinatorial disk (a single simplex),
while the output complex is a combinatorial annulus (a disk with a hole).
More precisely,
the input complex is a single $n$-simplex $\sigma^n$ properly colored with
$[n]$ and its faces.
For brevity, we use $\sigma^n$ to refer to this complex.
Let $\langle i_0 i_1 \hdots i_j \rangle$ denote the oriented face of $\sigma^n$
with colors $i_0, i_1, \hdots, i_j$ and with the orientation that contains the 
sequence $\langle i_0 i_1 \hdots i_j \rangle$.
Clearly, $\cC(\sigma^n)$ is a $\cS_n$-chain complex:
for each $\pi \in \cS_n$, 
$\pi(\ang{i_0 i_1 \hdots i_j}) = \langle \pi(i_0) \pi(i_1) \hdots \pi(i_j) \rangle$.

The output complex $\cA^n$ is defined as follows.
Each vertex has the form $(P_i,b_i)$,
where $P_i$ is a process ID and $v_i$ is 0 or 1.
A set of vertexes $\set{(P_0, v_0), \ldots, (P_j, v_j)}$ defines a
simplex of $\cA^n$ if the $P_i$ are distinct,
and if $j=n$ then the $b_i$ are not all 0 or all 1.
This complex is an \emph{annulus} (Figure~\ref{fig:annulus}).
Clearly, that $\cC(\cA^n)$ is a $\cS_n$-chain complex: for each $\pi \in \cS_n$, 
$\pi(\ang{({P_0, b_0}) \ldots (P_j, b_j)}) = \ang{(\pi(P_0), b_0) \hdots (\pi(P_j), b_j)}$.

\section{An Equivariance Theorem}
As explained below,
the existence of a protocol for WSB is tied to the existence of
an equivariant chain map from the disk to the annulus.
\begin{theorem}\label{theo}
There exists a non-trivial color-preserving equivariant chain map
$a: \cC(\sigma^n) \to \cC(\cA^n)$ 
if and only if the binomial coefficients
${n+1 \choose 1}, \hdots, {n+1 \choose n}$ are relatively prime.
\end{theorem}

\subsection{Necessity}
\label{sec:necessity}
In this section we prove the ``only if" direction:
if ${n+1 \choose 1}, \hdots, {n+1 \choose n}$ are not relatively prime,
there is no non-trivial color-preserving equivariant chain map 
$a:\cC(\sigma^n) \to \cC(\cA^n)$.
We prove that $a$  must map the boundary $\partial \sigma^n$ to a cycle
of $\cC(\cA^n)$ that is not a boundary,
a contradiction since chain maps preserve cycles and boundaries.

Consider the chain map $z: \cC(\bd(\sigma^n)) \to \cC(\cA^n)$
that maps each simplex $\ang{c_0 \hdots c_i}$ of $\cC(\bd(\sigma^n))$ to
$\ang{(c_0, 0) \ldots (c_i, 0)}$.
This map is color-preserving and equivariant. 
By induction on the dimension of the faces of $\sigma^n$, it can be proved the 
following lemma.

\begin{lemma}
\label{lemma1}
For each subset $s$ of $[n]$ there are families of equivariant homomorphisms
\begin{align*}
d^s_q: \cC^q(\sigma^n)	&\to \cC^{q+1}(\cA^n)	\\
f^s_p: \cC^p(\sigma^n)	&\to \cC^p(\cA^n)
\end{align*}
for $-1 \leq q \leq n-2$ and $0 \leq p \leq n-1$.
Moreover, for any proper $q$-dimensional face $\sigma$ of $\sigma^n$, the chain 
\begin{equation*}
a(\sigma) - z(\sigma) - d^{\ids(\sigma)}(\partial \sigma) -
\sum_{\sigma' \in \sk^{q-2}(\sigma)} f^{\ids(\sigma')}(\sigma)
\end{equation*}
is a $q$-cycle.
\end{lemma}

%Consider a proper $q$-dimensional face $\sigma$ of $\sigma^n$.
%Roughly speaking, in the cycle 
%$\alpha = a(\sigma) - z(\sigma) - d^{\ids(\sigma)}(\partial \sigma) - \sum_{\sigma' \in \sk^{q-2}(\sigma)} f^{\ids(\sigma')}(\sigma)$,
%$d^{\ids(\sigma)}(\partial \sigma)$ is what the $(q-1)$-dimensional faces of $\sigma$ 
%add for $\alpha$ and $\sum_{\sigma' \in \sk^{q-2}(\sigma)} f^{\ids(\sigma')}(\sigma)$ is what 
%the $\ell$-dimensional faces of $\sigma$,
%$\ell \leq q-2$, add for $\alpha$.

Let $\partial 0^n$ be the $(n-1)$-cycle of $\cC(\cA^n)$ defined as 
$\sum^n_{i=0} (-1)^i \ang{(P_0,0) \ldots \widehat{(P_i,0)} \ldots (P_n,0)}$, 
where circumflex ($\, \, \widehat{ } \, \,$) denotes omission.
Notice that $z(\partial \sigma^n) = \partial 0^n$ and $\partial 0^n$ is not a boundary.
Using Lemma \ref{lemma1} we can prove Theorem \ref{theo2}.
Informally, this theorem says that if the coefficients are not relatively prime,
any such map is forced to wrap non-zero ``times" $\partial \sigma^n$,
the boundary of a ``solid region'' $\sigma^n$,
around $0^n$, the boundary of a ``hole'' in $\cA^n$.
Because the map in question is a chain map sending boundaries to boundaries, 
it cannot exist.

\begin{theorem}
\label{theo2}
Let $a: \cC(\sigma^n) \to \cC(\cA^n)$ be a non-trivial color-preserving equivariant chain map.
For some set of integers $k_0, \ldots, k_{n-1}$,
\begin{equation*}
a(\partial \sigma^n) \thicksim
\left( 1 + \sum^{n-1}_{q=0} k_q {n+1 \choose q+1} \right) \partial 0^n.
\end{equation*}

\end{theorem}

\begin{proof}
(Sketch)
Let $\sigma_i$ denote the $(n-1)$-dimensional face $\ang{0 \hdots \widehat i \hdots n}$ of $\sigma^n$.
By Lemma \ref{lemma1},
\begin{equation*}
\alpha_i = a(\sigma_i) - z(\sigma_i) - d^{\ids(\sigma_i)}(\partial \sigma_i) -
\sum_{\sigma' \in \sk^{n-3}(\sigma_i)} f^{\ids(\sigma')}(\sigma_i)
\end{equation*}
is an $(n-1)$-cycle. 
Because $a$, $z$, $d$ and $f$ are equivariant,
for every $i \in [n]$,
$\alpha \thicksim (-1)^i k_{n-1} \partial 0^n$
for some integer $k_{n-1}$.
Therefore, $\sum^q_{i=0} (-1)^i \alpha_i \thicksim k_{n-1}(n+1)\partial 0^n$, hence 
\begin{equation}
\label{eqsketch1}
a(\partial \sigma^n) \thicksim (1 + k_{n-1} (n+1)) \partial 0^n + \gamma + \lambda
\end{equation}
where
\begin{equation*}
\gamma = \sum^{n}_{i=0} (-1)^i d^{\ids(\sigma_i)}(\partial \sigma_i) \hspace{1.5cm} \hbox{and} \hspace{1.5cm}
\lambda = \sum^{n}_{i=0} (-1)^i \sum_{\sigma' \in \sk^{n-3}(\sigma_i)}
f^{\ids(\sigma')}(\sigma_i),
\end{equation*}
since $z(\partial \sigma^n) = \partial 0^n$.

It is not hard to check that $\gamma = \sum^n_{i=0} \sum^n_{j+1} \alpha_{ij}$, where 
$\alpha_{ij} = (-1)^{i+j} ( d^{\ids(\sigma_j)}(\sigma_{i j}) - d^{\ids(\sigma_i)}(\sigma_{i j}) )$
and $\sigma_{i j}$ is the $(n-2)$-dimensional face
$\ang{0 \hdots \widehat i \hdots \widehat j \hdots n}$ of $\sigma^n$. 
It can be proved that $\alpha_{ij}$ is an $(n-1)$-cycle.
Moreover, using the fact that $a$, $z$, $d$ and $f$ are equivariant, 
we can prove that $\alpha_{ij} \thicksim k_{n-2} \partial 0^n$ for some integer $k_{n-2}$, for every $0 \leq i < j \leq n$.
Therefore, 
\begin{equation}
\label{eqsketch2}
\gamma \thicksim {n+1 \choose n-1} k_{n-2} \partial 0^n
\end{equation}

We can prove that $\lambda = \sum_{\sigma \in \sk^{n-3}(\sigma^n)} \alpha_\sigma$,
where $\alpha_\sigma = \sum_{i \in [n] - \ids(\sigma)} (-1)^i f^{\ids(\sigma)}(\sigma_i)$. 
Moreover, each $\alpha_\sigma$ is an $(n-1)$-cycle.
As for $\gamma$, it can be proved that for each $\sigma \in \sk^{n-3}(\sigma^n)$ of 
dimension $q$, $\alpha_\sigma \thicksim k_q \partial 0^n$, for some integer $k_q$. 
Thus, 
\begin{equation}
\label{eqsketch3}
\lambda \thicksim \sum^{n-3}_{i=0} {n+1 \choose i+1} k_q \partial 0^n
\end{equation}

The theorem follows from Equations (\ref{eqsketch1}), (\ref{eqsketch2}) and (\ref{eqsketch3}).
\end{proof}

Theorem \ref{theo2} says that 
$a(\partial \sigma^n) \thicksim ( 1 + \sum^{n-1}_{q=0} k_q {n+1 \choose q+1} ) \partial 0^n$.
It follows from elementary Number Theory that 
if ${n+1 \choose 1}, \hdots$, ${n+1 \choose n}$ are not relatively prime,
this equation has no integer solutions,
implying that $a(\partial \sigma^n)$ is not a boundary, hence $a$ cannot exist.

\begin{lemma}
\label{lemma2}
If the binomial coefficients ${n+1 \choose 1}, \hdots, {n+1 \choose n}$ are not relatively prime then 
there is no non-trivial color-preserving equivariant chain map $a: \cC(\sigma^n) \to \cC(\cA^n)$.
\end{lemma}

\subsection{Sufficiency}
\label{sec:sufficiency}
In this section we prove the ``if" direction:
if ${n+1 \choose 1}, \hdots, {n+1 \choose n}$ are relatively prime,
then there is a non-trivial color-preserving equivariant chain map 
$a:\cC(\sigma^n) \rightarrow \cC(\cA^n)$.

Earlier work~\cite{crpodc08} presents a construction that takes a simplex
$\sigma^n$ and a set of integers $\{ k_0, \hdots, k_{n-1} \}$ with  $k_0 \in \{ 0, -1 \}$,
and produces a subdivision $\chi(\sigma^n)$ with the following two colorings.
First, $\ids$ is a proper coloring with respect to $[n]$.
Second, $b$ is a binary coloring which induces
$1 + \sum^{n-1}_{i=0} k_i {n+1 \choose i+1}$ monochromatic $n$-simplexes.
The binary coloring $b$ is symmetric in a sense that for each pair of
$m$-faces $\sigma_i$ and $\sigma_j$ of $\sigma^n$,
there is a simplicial bijection $\mu_{ij}: \chi(\sigma_i) \rightarrow \chi(\sigma_j)$
such that for every vertex $v \in \chi(\sigma_i)$, 
$b(v) = b(\mu(v))$ and $\rk(\ids(v)) = \rk(\ids(\mu(v)))$,
where $\rk: \ids(\sigma_i) \to \ids(\sigma_j)$ is the rank function
such that if $a < b$ in $\ids(\sigma_i)$, then $\rk(a) < \rk(b)$.

By a standard construction, subdivisions induce chain maps. 
In particular, $\chi(\sigma^n)$ induces a chain map
$\mu_1: \cC(\sigma^n) \to \cC(\chi(\sigma^n))$.
The colorings $id$ and $b$ define a simplicial map
$\chi(\sigma^n) \to \cA^n$ only if $b$ defines no monochromatic
$n$-simplexes in $\chi(\sigma^n)$.
Specifically, if $1 + \sum^{n-1}_{i=0} k_i {n+1 \choose i+1} = 0$.
It follows from elementary Number Theory that
that if ${n+1 \choose 1}, \hdots$, ${n+1 \choose n}$ are relatively prime,
then the equation ${n+1 \choose 1} k_0 + {n+1 \choose 2} k_2 + \hdots + {n+1 \choose n} k_{n-1} = 1$
has an integer solution, thus the simplicial map induced by $id$ and $b$ induces a chain map
$\mu_2: \cC(\chi(\sigma^n)) \to \cC(\cA^n)$.

Let $a$ be the composition $\mu_2 \circ \mu_1$.
Since $\chi(\sigma^n)$ is a chromatic subdivision of $\sigma^n$,
$a$ is clearly non-trivial and color-preserving. 
To show that $a$ is equivariant,
one can prove by induction on $q$ that the restriction
$a \vert_{\cC(\sk^q(\sigma^n))}$, $0 \leq q \leq n$, is equivariant.
By symmetry of $b$, the base case $q=0$ is trivial.
For the induction hypothesis, assume $a \vert_{\cC(\sk^{q-1}(\sigma^n))}$ is equivariant.
The induction step consists in proving that, for each $q$-face
$\sigma = \and{c_0 \hdots c_q}$ of $\sigma^n$,
$a(\partial \sigma)$ forces the value $a(\sigma)$ such that $\pi \circ a(\sigma) = a \circ \pi(\sigma)$
for every $\pi \in \cS_n$, hence $a \vert_{\cC(\sk^q(\sigma^n))}$ is equivariant.
Roughly speaking, the proof first observes that 
$a(\sigma) = \sum_{\tau \in L_q} k_{\tau} \tau$,
where $L_q = \{ \tau \vert \tau \in \cA^n \hbox{ and } \ids(\tau) = \ids(\sigma) \}$ 
and integer $k_{\tau}$, since $a$ is color preserving.
The induction hypothesis that $a \vert_{\cC(\sk^{q-1}(\sigma^n))}$ is equivariant
implies that $a(\partial \sigma)$ forces the value $a(\sigma)$
such that $k_\tau = k_{\tau'}$ for 
$\tau, \tau' \in L_{q, k} = \{ \tau \vert \tau \in L_q \hbox{ and } \vert \set{ v \in \tau \vert b(v) = 1 } \vert = k \}$,
$0 \leq k \leq q+1$.
For example, for $\ang{012}$ and $k=2$, this says that if 
$\ang{(0,0) (1,1) (2,1)}$ appears in $a(\ang{012})$ with coefficient $\ell$,
then $\ang{(0,1) (1,1) (2,0)}$ and $\ang{(0,1) (1,0) (2,1)}$ appear in $a(\ang{012})$
with coefficient $\ell$ too. It is not hard to see that this proves 
$a \circ \pi(\sigma) = \pi \circ a(\sigma)$ for every $\pi \in \cS_n$, 
hence the inductive step is done. 

\begin{lemma}
\label{lemma3}
If the binomial coefficients ${n+1 \choose 1}, \hdots, {n+1 \choose n}$ are
relatively prime then there is a non-trivial color-preserving equivariant chain map $a:
\cC(\sigma^n) \to \cC(\cA^n)$.
\end{lemma}

\section{Applications to Distributed Computing}
Theorem~\ref{theo} is a statement about the existence of equivariant chain maps
between two simple topological spaces.
In this section, we explain what this theorem says about distributed computing.

Informally,
a complex is \emph{$k$-connected} if any continuous map from the boundary of a
$k$-simplex to the complex can be extended to a continuous map of the entire simplex.
It is known that if a protocol complex $k$-connected,
then it cannot solve $(k+1)$-set agreement~\cite{hrmscs00,hsjacm99}.
In the \emph{$(k+1)$-set agreement} the processes start with a private input 
value and each chooses an output value among input values;
at most $k+1$ distinct output values are elected.

Here is how to apply this theorem to tell if there is no wait-free protocol
for $2n$-renaming for $(n+1)$ processes in wait-free read-write memory.
This description is only a summary:
the complete construction appears elsewhere~\cite{hrmscs00}.
Recall that WSB and $2n$-renaming are equivalent in 
an asynchronous system made of $n+1$ processes that 
communicate using a read/write shared memory~\cite{grhdisc06}.

The WSB task is given by $(\sigma^n, \cA^n, \Delta)$,
where $\sigma^n$ is a properly colored simplex that represents the
unique input configuration, $\cA^n$ is the annulus corresponding to all
possible output binary values, 
and $\Delta(\sigma^n)$ defines all legal assignments.
Assume we have a wait-free protocol $\cP$ that solves WSB,
and let $\cP(\sigma^n)$ be the complex generated by all executions of the
protocol starting from $\sigma^n$.
Any such protocol complex is $n$-connected~\cite{hsjacm99}.

The anonymity requirement for WSB induces a symmetry on the 
binary output values of the boundary of $\cP(\sigma^n)$.
This symmetry allows to construct a an equivariant simplicial map
$\phi: \cP(\sigma^n) \to \cP(\cA^n)$.
Prepending the map $\cC(\sigma^n) \to \cC(\cP(\sigma^n))$ induced by a subdivision,
this equivariant simplicial map induces equivariant chain maps:
\begin{equation*}
\begin{CD}
\cC(\sigma^n) @>>>
\cC(\cP(\sigma^n)) @>>>
\cC(\cA^n).
\end{CD}
\end{equation*}
The composition of these maps yields an equivariant chain
map $a: \cC(\sigma^n) \to \cC(\cA^n)$.
Theorem~\ref{theo}, however, states that this chain map does not exist if the
binomial coefficients are not relatively prime.
\begin{corollary}
if  ${n+1\choose 1}, \hdots, {n+1 \choose n}$ are not relatively prime,
there is no wait-free $2n$-renaming protocol in the asynchronous read/write memory
or message-passing models.
\end{corollary}
There is a protocol if the coefficients are  relatively prime~\cite{crpodc08},
but that claim is not implied by this corollary.

\begin{figure}
\centerline{\input{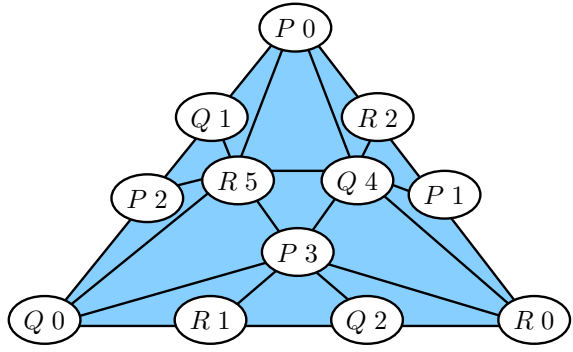}}
\caption{Symmetric input subcomplex for renaming}
\label{figure:symmetric}
\end{figure}

In the more general case, where $t$ out of $n+1$ processes can fail,
the construction is a bit more complicated and 
the dimensions shrink~\cite{exBG}.
The $2n$-renaming task is given by $(\cI,\cO,\Delta)$,
where $\cI$ is the complex defining all possible input name assignments,
$\cO$ is all possible assignments of output names taken from $0, \ldots, 2n-1$,
and for each $\sigma^n \in \cI$,
$\Delta(\sigma^n)$ defines all legal name assignments.

Assume we have a $t$-resilient $(n+t)$-renaming protocol.
Partition the set of processes into two sets,
$n-t$ \emph{passive} processes, and $t+1$ \emph{active} processes.
If $\cC$ is a complex labeled with process IDs,
let $\cC_a$ be the subcomplex labeled with IDs of active processes.
Let $\cP^{*}$ be the protocol complex for executions in which none of the
passive processes fail,
so all failures are distributed among the active processes.
As illustrated in Figure~\ref{figure:symmetric},
we can identify a subcomplex of $\cI$ isomorphic to a subdivision
$\chi(\sigma^n)$ of an $n$-simplex $\sigma^n$,
where the input names are symmetric along the boundary.
Because $\cP_a^{*}(\chi(\sigma^n))$ is $t$-connected~\cite{HerlihyR10}
and by the anonymity requirement for renaming,
we can construct a simplicial map $\phi: \chi^N(\sigma^t) \to \cP_a^{*}(\chi(\sigma^n))$
from a subdivision of a $t$-simplex $\sigma^t$ to the subcomplex of the restricted
protocol complex labeled with active IDs.
The simplicial map $\phi$ is equivariant under $\cS_{t+1}$,
the symmetry group acting on the active process IDs,
as is the simplicial decision map $\delta: \cP^* \to \cO$.
It follows that every passive process takes the same output name in every
execution of $\cP^*$.
Without loss of generality, assume these passive names are $2t, \ldots, n+t-1$,
leaving the range $0, \ldots, 2t-1$ to the active processes.
Let $\pi: \cO_a \to \cA^t$ send each remaining name to its parity.

These equivariant simplicial maps form a sequence:
\begin{equation*}
\begin{CD}
\chi^N(\sigma^t) @>\phi>>
\cP_a^*(\chi(\sigma^n)) @>\delta>>
\cO_a @>\pi>>
\cA^t,
\end{CD}
\end{equation*}
which induces the following sequence of chain maps:
\begin{equation*}
\begin{CD}
\cC(\sigma^t) @>>>
\cC(\cP^*_a(\chi(\sigma^n))) @>>>
\cC(\cO_a) @>>>
\cC(\cA^t).
\end{CD}
\end{equation*}
The composition of these maps yields an equivariant chain
map $a: \cC(\sigma^t) \to \cC(\cA^t)$.
Theorem~\ref{theo}, however, states that this map does not exist if the
binomial coefficients are not relatively prime.
\begin{corollary}
if  ${t+1\choose 1}, \hdots, {t+1 \choose t}$ are not relatively prime,
there is no $t$-resilient $(n+t)$-renaming protocol in the asynchronous read-write
memory or message-passing models.
\end{corollary}
It is unknown whether there is a protocol if the coefficients are relatively prime.

\appendix

\newpage

\section{Proofs}

\subsection{Proofs of Section \ref{sec:necessity}}

For distinct $i_0, i_1, \hdots, i_q \in [n]$, $q \leq n-1$,
let ${\cal S}^q_{i_0 i_1 \hdots i_q}$ denote the subcomplex of ${\cal A}^n$
that contains all $q$-simplexes, and all its faces, that are properly colored with $i_0, i_1, \hdots, i_q$.
It is not hard to see that ${\cal S}^q_{i_0 i_1 \hdots i_q}$ is a sphere of dimension $q$.  

\begin{lemma}
\label{lemmaspheres}
Let ${\cal S}$ be a sphere of dimension $n$. 
Then every $\ell$-cycle is a boundary, $\ell \leq n-1$.
\end{lemma}

%\begin{lemma}[\cite{hsjacm99}]
%\label{lemmahsjacm1}
%Let $\cA^n$ be the annulus for dimension $n$. 
%Then every $\ell$-cycle is a boundary, $\ell \leq n-2$.
%\end{lemma}

\begin{lemma}[\cite{hsjacm99}]
\label{lemmahsjacm2}
Let $S_i$ be the cycle obtained by orienting each $(n-1)$-simplex of $\cS^{n-1}_{0 \hdots \widehat i \hdots n}$.
Then, every $(n-1)$-cycle of ${\cal C}({\cal O}^n) $ is homologous to $k \hbox{ } S$ for some integer $k$.
\end{lemma}

\begin{lemma}[\cite{hsjacm99}]
\label{lemmahsjacm3}
Let $S_i$ be the cycle obtained by orienting the $(n-1)$-simplexes of 
$\cS^{n-1}_{0 \hdots \widehat i \hdots n}$ such that its $0$-monochromatic $(n-1)$-simplex is oriented
in increasing $\ids$ order. Then, $S_i \thicksim (-1)^i \partial 0^n$.
\end{lemma}

In what follows, for $0 \leq i \leq m \leq n$, let $\pi^m_i$ denote the permutation defined as follows:
$$
\pi^m_i = \left(
\begin{array}{cccccccccc}
0 & \hdots & i-1 & i & \hdots & m-1 & m & m+1 &  \hdots & n\\
0 & \hdots & i-1 & i+1 & \hdots & m & i & m+1 &  \hdots & n
\end{array}
\right)
$$

%\vspace{0.2cm}
\noindent
\emph{Lemma \ref{lemma1} (Restated)}
For each subset $s$ of $[n]$ there are families of equivariant homomorphisms
\begin{align*}
d^s_q: \cC^q(\sigma^n)	&\to \cC^{q+1}(\cA^n)	\\
f^s_p: \cC^p(\sigma^n)	&\to \cC^p(\cA^n)
\end{align*}
for $-1 \leq q \leq n-2$ and $0 \leq p \leq n-1$.
Moreover, for any proper $q$-dimensional face $\sigma$ of $\sigma^n$, the chain 
\begin{equation*}
a(\sigma) - z(\sigma) - d^{\ids(\sigma)}(\partial \sigma) -
\sum_{\sigma' \in \sk^{q-2}(\sigma)} f^{\ids(\sigma')}(\sigma)
\end{equation*}
is a $q$-cycle.

\noindent
\begin{proof}
We proceed by induction on the dimension of the faces of $\sigma^n$.
Unless stated otherwise, $d^s = 0$ and $f^s = 0$.
For the rest of the proof let $\sigma_{i_0 i_1 \hdots i_j}$ denote the 
oriented face $\ang{i_0 i_1 \hdots i_j}$ of $\sigma^n$.

For dimension $0$ it is easy to see that, for each $0$-face $\sigma$ of $\sigma^n$,
 $a(\sigma) - z(\sigma)$ is a $0$-cycle.
For dimension $1$, consider the face $\sigma_0$ and the set $\{ 0, 1 \}$.
We have that $a(\sigma_0) - z(\sigma_0)$ is a $0$-cycle.
Moreover, since $a$ is color-preserving and by the definition of $z$, 
$a(\sigma_0), z(\sigma_0) \in \cC(\cS^0_{0})$.
By Lemma \ref{lemmaspheres} and since $\cS^0_{0} \subset \cS^1_{01}$, 
there is a $1$-chain $d^{01}(\sigma_0) \in \cC(\cS^1_{01})$ 
such that $\partial d^{01}(\sigma_0) = a(\sigma_0) - z(\sigma_0)$.
Now, using $d^{01}(\sigma_0)$, we ``symmetrically" define the value of
$d$ for each pair of $0$-face $\sigma$ and set $s$ of size $2$ such that
$\ids(\sigma) \subset s$, namely, $d^{\pi(01)}(\pi(\sigma_0)) = d^s(\sigma) = \pi(d^{01}(\sigma_0))$,
where $\pi$ is a permutation such that $\sigma = \pi(\sigma_0)$ and $s = \pi(01)$.
In this way
\begin{eqnarray*}
\partial d^s(\sigma) &=&  \partial \pi(d^{01}(\sigma_0)) = \pi(\partial d^{01}(\sigma_0))\\
&=& \pi(a(\sigma_0) - z(\sigma_0)) = a(\pi(\sigma_0)) - z(\pi(\sigma_0)) = a(\sigma) - z(\sigma)
\end{eqnarray*}
Observe that $d^s(\sigma) \in \cC(\cS^1_s)$.

For example, for the $0$-face $\sigma_1$, 
$d^{\pi^1_0(01)}(\pi^1_0(\sigma_0)) = d^{01}(\sigma_1) = \pi^1_0(d^{01}(\sigma_0))$.
%Notice that $d^{01}(\sigma_1) \in \cC(\cS^1_{01})$.
Observe that the election of $d^{01}(\sigma_0)$ allows to achieve an equivariant $d$. 
Thus, we have that $\partial d^{01}(\sigma_0) = a(\sigma_0) - z(\sigma_0)$ and
$\partial d^{01}(\sigma_1) = a(\sigma_1) - z(\sigma_1)$, hence
\begin{eqnarray*}
\partial d^{01}(\sigma_1) - \partial d^{01}(\sigma_0) &=& 
a(\sigma_1) - z(\sigma_1) - (a(\sigma_0) - z(\sigma_0)) \\
\partial d^{01}(\partial \sigma_{01}) &=& a(\partial \sigma_{01}) - z(\partial \sigma_{01})\\
0 &=& \partial( a(\sigma_{01}) - z(\sigma_{01}) - d^{01}(\partial \sigma_{01}) )
\end{eqnarray*}
Thus, $a(\sigma_{01}) - z(\sigma_{01}) - d^{01}(\partial \sigma_{01})$
is a $1$-cycle. This complete the basis of the induction,
however we present the case for dimension $2$ to illustrate the idea.

Consider the face $\sigma_{01}$ and the set $\{ 0, 1, 2 \}$.
We have proved that $a(\sigma_{01}) - z(\sigma_{01}) - d^{01}(\partial \sigma_{01})$
is a $1$-cycle. Moreover, since $a$ and $z$ are color-preserving, and by the previous step,
we have that $a(\sigma_{01}), z(\sigma_{01}), d^{01}(\partial \sigma_{01}) \in \cC(\cS^1_{01})$.
By Lemma \ref{lemmaspheres} and since $\cS^1_{01} \subset \cS^2_{012}$, 
there exists a $2$-chain $d^{012}(\sigma_{01}) \in \cC(\cS^2_{012})$ such that 
$\partial d^{012}(\sigma_{01}) = a(\sigma_{01}) - z(\sigma_{01}) - d^{01}(\partial \sigma_{01})$.
Using $d^{012}(\sigma_{01})$, we ``symmetrically" define the value of
$d$ for all pair of $1$-face $\sigma$ and set $s$ of size $3$ such that $\ids(\sigma) \subset s$.
For example, $d^{\pi^2_0(012)}(\pi^2_0(\sigma_{01})) = d^{012}(\sigma_{12}) = \pi^2_0(d^{012}(\sigma_{01}))$
and $d^{\pi^2_1(012)}(\pi^2_1(\sigma_{01})) = d^{012}(\sigma_{02}) = \pi^2_1(d^{012}(\sigma_{01}))$.
So we have 
\begin{eqnarray*}
\partial d^{012}(\sigma_{01}) &=& a(\sigma_{01}) - z(\sigma_{01}) - d^{01}(\partial \sigma_{01})\\
\partial d^{012}(\sigma_{12}) &=& a(\sigma_{12}) - z(\sigma_{12}) - d^{12}(\partial \sigma_{12})\\
\partial d^{012}(\sigma_{02}) &=& a(\sigma_{02}) - z(\sigma_{02}) - d^{02}(\partial \sigma_{02})
\end{eqnarray*}
Taking the alternating sign sum over $\sigma_{01}, \sigma_{12}, \sigma_{02}$,
\begin{eqnarray*}
\partial d^{012}(\sigma_{01}) - d^{012}(\sigma_{02}) + d^{012}(\sigma_{12}) 
&=& + (a(\sigma_{01}) - z(\sigma_{01}) - d^{01}(\partial \sigma_{01}))\\
& & - (a(\sigma_{02}) - z(\sigma_{02}) - d^{02}(\partial \sigma_{02}))\\
& & + (a(\sigma_{12}) - z(\sigma_{12}) - d^{12}(\partial \sigma_{12}))\\
\partial d^{012}(\partial \sigma_{012}) 
&=& a(\partial \sigma_{012}) - z(\partial \sigma_{012}) - \gamma
\end{eqnarray*}
where $\gamma = d^{12}(\partial \sigma_{12}) - d^{02}(\partial \sigma_{02}) + d^{01}(\partial \sigma_{01})$.
Thus  
\begin{equation}
\label{eq1n=2}
\partial( a(\sigma_{012}) - z(\sigma_{012}) - d^{012}(\partial \sigma_{012})) - \gamma = 0
\end{equation}

Now, we have that 
\begin{eqnarray*}
\gamma &=& d^{12}(\partial \sigma_{12}) - d^{02}(\partial \sigma_{02}) + d^{01}(\partial \sigma_{01}))\\
&=& (d^{12}(\sigma_2) - d^{12}(\sigma_1)) - (d^{02}(\sigma_2) - d^{02}(\sigma_0))
+ (d^{01}(\sigma_1) - d^{01}(\sigma_0))
\end{eqnarray*}
Considering the result of the boundary operator over the terms where $\sigma_0$ appears, we get
\begin{eqnarray*}
\partial( d^{02}(\sigma_0) - d^{01}(\sigma_0) ) &=& 
\partial d^{02}(\sigma_0) - \partial d^{01}(\sigma_0)\\
&=& a(\sigma_0) - z(\sigma_0) - ( a(\sigma_0) - z(\sigma_0) )\\ 
&=& 0
\end{eqnarray*}
Thus, $d^{02}(\sigma_0) - d^{01}(\sigma_0)$ is a $1$-cycle. 
The same happens with the terms where $\sigma_1$ and $\sigma_2$ appear, respectively.
Now, by Lemma \ref{lemmaspheres} and since $d^{01}(\sigma_0)\in \cC(\cS^1_{01})$ and
$d^{02}(\sigma_0) \in \cC(\cS^1_{02})$,
there is a $2$-chain $f^0(\sigma_{012}) \in \cC(\cS^2_{012})$
such that $\partial f^0(\sigma_{012}) = d^{02}(\sigma_0) - d^{01}(\sigma_0)$.
The value $f^0(\sigma_{012})$ induces the value of $f$ for all pair of 
$2$-face $\sigma$ and set $s$ of size $1$ such that $s \subset \ids(\sigma)$.
For example, $f^{\pi^2_0(0)}(\pi^2_0(\sigma_{012})) = f^1(\sigma_{120}) = f^1(\sigma_{012}) = \pi^2_0(f^0(\sigma_{012}))$.
Observe that $f^1(\sigma_{012}), f^2(\sigma_{012}) \in \cC(\cS^2_{012})$.
Therefore,
\begin{equation}
\label{eq2n=2}
\gamma = \partial f^0(\sigma_{012}) + \partial f^1(\sigma_{012}) + \partial f^2(\sigma_{012})
= \partial \sum_{ \sigma \in \sk^0(\sigma_{012}) } f^{\ids(\sigma)}(\sigma_{012}) 
\end{equation}
Combining equations (\ref{eq1n=2}) and (\ref{eq2n=2}) we get
$$0 = \partial \left( a(\sigma_{012}) - z(\sigma_{012}) - d^{012}(\partial \sigma_{012})  - 
\sum_{ \sigma \in \sk^0(\sigma_{012}) } f^{\ids(\sigma)}(\sigma_{012}) \right)$$
hence the lemma holds for $n=2$.
Roughly speaking, $f^i(\sigma_{012})$, $i \in \{ 0, 1, 2 \}$, is what the $0$-dimensional face 
$\sigma_i$ of $\sigma_{012}$ adds in obtaining the $2$-cycle for $\sigma_{012}$.

Assume the lemma holds for faces of dimension $q-1$, $0 \leq q \leq n-1$. 
We prove the lemma holds for faces of dimension $q$.
Also, for each $(q-1)$-dimensional face $\sigma = \sigma_{c_0 \hdots c_{q-1}}$, assume the following.

\begin{enumerate}

\item For every $(q-2)$-dimensional face $\sigma'$ of $\sigma$,
$d^{\ids(\sigma)}(\sigma') \in \cC(\cS^{q-1}_{\ids(\sigma)})$,
and for each $\ell$-dimensional face $\sigma'$ of $\sigma$, $\ell \leq q-3$,
$f^{\ids(\sigma')}(\sigma) \in \cC(\cS^{q-1}_{\ids(\sigma)})$.

\item For every $(q-2)$-dimensional face $\sigma'$ of $\sigma$,
$$\partial d^{\ids(\sigma)}(\sigma') = 
a(\sigma') - z(\sigma') - d^{\ids(\sigma')}(\partial \sigma') - 
\sum_{\sigma'' \in \sk^{q-4}(\sigma') } f^{\ids(\sigma'')}(\sigma')$$

\item For every $(q-3)$-dimensional face $\sigma' = \sigma_{c_0 \hdots \widehat{c_i} \hdots\widehat{c_j} \hdots c_{q-1}}$ 
of $\sigma$,
$$\partial f^{\ids(\sigma')}(\sigma) = (-1)^{i+j} (d^{\ids(\sigma_j)}(\sigma') - d^{\ids(\sigma_i)}(\sigma'))$$
where $\sigma_i = \sigma_{c_0 \hdots \widehat{c_i} \hdots c_{q-1}}$ and 
$\sigma_j = \sigma_{c_0 \hdots \widehat{c_j} \hdots c_{q-1}}$.

\item For every $k$-dimensional face $\sigma' $ of $\sigma$, $k \leq q-4$,
$$\partial f^{\ids(\sigma')}(\sigma) =  
\sum_{c_i \in \ids(\sigma), c_i \notin \ids(\sigma')} 
(-i)^i f^{\ids(\sigma')}(\sigma_i)$$
where $\sigma_i = \sigma_{c_0 \hdots \widehat{c_i} \hdots c_{q-1}}$.
\end{enumerate}

Consider the $q$-simplex $\sigma = \sigma_{0 \hdots q}$.
Let $\sigma_i$ be the $(q-1)$-dimensional face $\sigma_{0 \hdots \widehat{i} \hdots q}$ of $\sigma$.
By induction hypothesis,
$$a(\sigma_i) - z(\sigma_i) - d^{\ids(\sigma_i)}(\partial \sigma_i) - 
\sum_{\sigma' \in \sk^{q-3}(\sigma_i)} f^{\ids(\sigma')}(\sigma_i)$$
is a $(q-1)$-cycle. Consider the $(q-1)$-dimensional face $\sigma_q$.
By induction hypothesis, for each $(q-2)$-dimensional face $\sigma'$ of $\sigma_q$, 
$d^{0 \hdots q-1}(\sigma') \in \cC(\cS^{q-1}_{0 \hdots q-1})$, 
and for each $\ell$-dimensional face $\sigma'$ of $\sigma_q$, $\ell \leq q-3$,
$f^{\ids(\sigma')}(\sigma_q) \in \cC(\cS^{q-1}_{0 \hdots q-1})$. 
Also, $a(\sigma_q), z(\sigma_q) \in \cC(\cS^{q-1}_{0 \hdots q-1})$,
because $a$ and $z$ are color-preserving.
By Lemma \ref{lemmaspheres} and since $\cS^{q-1}_{0 \hdots q-1} \subset \cS^q_{0 \hdots q}$, 
there is a $q$-chain $d^{0 \hdots q}(\sigma_q) \in \cC(\cS^q_{0 \hdots q})$
such that 
$$\partial d^{0 \hdots q}(\sigma_q) = a(\sigma_q) - z(\sigma_q) - d^{\ids(\sigma_q)}(\partial \sigma_q) - 
\sum_{\sigma' \in \sk^{q-3}(\sigma_q)} f^{\ids(\sigma')}(\sigma_q)$$
Using $d^{0 \hdots q}(\sigma_{q})$, we ``symmetrically" define the value of
$d^s(\sigma') = \pi(d^{0 \hdots q}(\sigma_{q}))$,
where $dim(\sigma') = q-1$, $\vert s \vert = q+1$,  $\ids(\sigma') \subset s$,
$\pi(\sigma_q) = \sigma'$ and $\pi(\{0, \hdots, q\}) = s$.
Therefore, for each face $\sigma_i$ of $\sigma$
$$\partial d^{0 \hdots q}(\sigma_i) = a(\sigma_i) - z(\sigma_i) - d^{\ids(\sigma_i)}(\partial \sigma_i) - 
\sum_{\sigma' \in \sk^{q-3}(\sigma_i)} f^{\ids(\sigma')}(\sigma_i)$$
and $d^{0 \hdots q}(\sigma_i) \in \cC(\cS^q_{0 \hdots q})$.

Taking the alternating sign sum over all $(q-1)$-faces of $\sigma$, we get
\begin{eqnarray*}
\sum^{q}_{i=0} (-1)^i \partial d^{0 \hdots q}(\sigma_i) &=&
\sum^{q}_{i=0} (-1)^i \left( a(\sigma_i) - z(\sigma_i) - d^{\ids(\sigma_i)}(\partial \sigma_i) - 
\sum_{\sigma' \in \sk^{q-3}(\sigma_i)} f^{\ids(\sigma')}(\sigma_i) \right)\\
\partial d^{0 \hdots q}(\partial \sigma) &=& a(\partial \sigma) - z(\partial \sigma) - 
\gamma - \lambda\\
0 &=& \partial ( a(\sigma) - z(\sigma) - d^{0 \hdots q}(\partial \sigma) ) - \gamma - \lambda
\end{eqnarray*}
where 
\begin{eqnarray*}
\gamma &=& \sum^{q}_{i=0} (-1)^i d^{\ids(\sigma_i)}(\partial \sigma_i)\\
\lambda &=& \sum^{q}_{i=0} (-1)^i \sum_{\sigma' \in \sk^{q-3}(\sigma_i)} f^{\ids(\sigma')}(\sigma_i)
\end{eqnarray*}
We now extend $d$ and $f$ such that 
\begin{equation}
\label{eq0induction}
\partial ( a(\sigma) - z(\sigma) - d^{\ids(\sigma)}(\partial \sigma) ) - \gamma - \lambda
\end{equation}
is a $q$-cycle. Intuitively, we will see that $\gamma$ and $\lambda$ are made of 
$(q-1)$-cycles, hence there are $q$-chains $\gamma'$ and $\lambda'$ such that 
$\partial \gamma' = \gamma$ and $\partial \lambda' = \lambda$. 
Combining $\partial \gamma'$ and $\partial \lambda'$ with Equation (\ref{eq0induction}), 
we get $a(\sigma) - z(\sigma) - d^{\ids(\sigma)}(\partial \sigma) - \gamma' - \lambda'$
is a $q$-cycle, since we know that 
$\partial ( a(\sigma) - z(\sigma) - d^{\ids(\sigma)}(\partial \sigma) ) - \gamma - \lambda = 0$.
As we shall see, $\gamma'$ and $\lambda'$ are the $q$-chains the lemma requires.

First, let us consider $\gamma$. Let $\sigma_{ij}$ denote
the $(q-2)$-dimensional face $\sigma_{0 \hdots \widehat i \hdots \widehat j \hdots q}$ of $\sigma$.
Observe that
\begin{eqnarray*}
\partial \gamma = \partial \sum^{q}_{i=0} (-1)^i d^{\ids(\sigma_i)}(\partial \sigma_i) &=& 
\partial \sum^{q}_{i=0} (-1)^i \left( \sum^{i-1}_{j=0} (-1)^j d^{\ids(\sigma_i)}(\sigma_{ji}) + 
\sum^q_{j=i+1} (-1)^{j-1} d^{\ids(\sigma_i)}(\sigma_{ij}) \right)\\
&=& \sum^q_{i=0} \sum^q_{j=i+1}  (-1)^{i+j} \partial \left( d^{\ids(\sigma_j)}(\sigma_{ij}) - 
d^{\ids(\sigma_i)}(\sigma_{ij}) \right)
\end{eqnarray*}
By induction hypothesis, $\partial d^{\ids(\sigma_j)}(\sigma_{ij}) = \partial d^{\ids(\sigma_i)}(\sigma_{ij})$,
thus the $(q-1)$-chain
$d^{\ids(\sigma_j)}(\sigma_{ij}) - d^{\ids(\sigma_i)}(\sigma_{ij})$
is a cycle.
In addition, $d^{\ids(\sigma_i)}(\sigma_{ij}) \in \cC(\cS^{q-1}_{\ids(\sigma_i)})$ and 
$d^{\ids(\sigma_j)}(\sigma_{ij}) \in \cC(\cS^{q-1}_{\ids(\sigma_j)})$, by induction hypothesis.
By Lemma \ref{lemmaspheres} and since 
$\cS^{q-1}_{\ids(\sigma_i)}, \cS^{q-1}_{\ids(\sigma_j)} \subset \cS^q_{0 \hdots q}$,
there exists a $q$-chain $f^{\ids(\sigma_{ij})}(\sigma) \in \cC(\cS^q_{0 \hdots q})$ such that 
$$\partial f^{\ids(\sigma_{ij})}(\sigma) = 
(-1)^{i+j} \left( d^{\ids(\sigma_j)}(\sigma_{ij}) - d^{\ids(\sigma_i)}(\sigma_{ij}) \right)$$
We use $f^{\ids(\sigma_{ij})}(\sigma)$ to ``symmetrically" define the value of 
$f^s(\sigma')$ for $dim(\sigma') = q$, $\vert s \vert = q-1$ and 
$s \subset \ids(\sigma')$. So we have 
\begin{equation}
\label{eq5induction}
\gamma = \partial \sum_{\sigma' \in \sk^{q-2}(\sigma), dim(\sigma') = q-2} f^{\ids(\sigma')}(\sigma)
\end{equation}

Consider now $\lambda$. It is not hard to see that  
$$\lambda = \sum^q_{i=0} (-1)^i \sum_{\sigma' \in \sk^{q-3}(\sigma_i)} f^{\ids(\sigma')}(\sigma_i) =
\sum_{\sigma' \in \sk^{q-3}(\sigma)} \sum_{i \in [q] - \ids(\sigma')} (-1)^i f^{\ids(\sigma')}(\sigma_i)$$
We prove that $\sum_{i \in [q] - \ids(\sigma')} (-1)^i f^{\ids(\sigma')}(\sigma_i)$ is a $(q-1)$-cycle.
Observe that $\sigma'$ is a face of $\sigma_i$.
Fix some $\sigma' \in \sk^{q-3}(\sigma)$.
We consider two cases, $dim(\sigma') = q-3$ and $dim(\sigma') \leq q-4$.

\paragraph{Case $dim(\sigma') = q-3$.}
Assume, without loss of generality, $[q] - \ids(\sigma') = \{ a, b, c \}$ with $a < b < c$.
We have that 
$$\partial \sum_{i \in [q] - \ids(\sigma')} (-1)^i f^{\ids(\sigma')}(\sigma_i) = 
(-1)^c \partial f^{\ids(\sigma')}(\sigma_c) + (-1)^b \partial f^{\ids(\sigma')}(\sigma_b) + (-1)^a \partial f^{\ids(\sigma')}(\sigma_a)$$
Let $\sigma_{ijk}$ denote the face $\sigma_{0 \hdots \widehat i \hdots \widehat j \hdots \widehat k \hdots q}$
of $\sigma$. By induction hypothesis, 
\begin{eqnarray*}
\partial f^{\ids(\sigma')}(\sigma_c) &=& (-1)^{a+b} f^{\ids(\sigma_{bc})}(\sigma_{abc}) + (-1)^{a+b-1} f^{\ids(\sigma_{ac})}(\sigma_{abc})\\
\partial f^{\ids(\sigma')}(\sigma_b) &=& (-1)^{a+c-1} f^{\ids(\sigma_{bc})}(\sigma_{abc}) + (-1)^{a+c-2} f^{\ids(\sigma_{ab})}(\sigma_{abc})\\
\partial f^{\ids(\sigma')}(\sigma_a) &=& (-1)^{b+c-2} f^{\ids(\sigma_{ac})}(\sigma_{abc}) + (-1)^{b+c-3} f^{\ids(\sigma_{ab})}(\sigma_{abc})  
\end{eqnarray*}
and thus
\begin{eqnarray*}
\partial \sum_{i \in [q] - \ids(\sigma')} (-1)^i f^{\ids(\sigma')}(\sigma_i) &=&  (-1)^{a+b+c} f^{\ids(\sigma_{bc})}(\sigma_{abc}) + (-1)^{a+b+c-1} f^{\ids(\sigma_{ac})}(\sigma_{abc})\\
& & + (-1)^{a+b+c-1} f^{\ids(\sigma_{bc})}(\sigma_{abc}) + (-1)^{a+b+c-2} f^{\ids(\sigma_{ab})}(\sigma_{abc})\\
& & + (-1)^{a+b+c-2} f^{\ids(\sigma_{ac})}(\sigma_{abc}) + (-1)^{a+b+c-3} f^{\ids(\sigma_{ab})}(\sigma_{abc}) \\
&=& 0
\end{eqnarray*}

Therefore, $\sum_{i \in [q] - \ids(\sigma')} (-1)^i f^{\ids(\sigma')}(\sigma_i)$ is a $(q-1)$-cycle.
By induction hypothesis, $f^{\ids(\sigma')}(\sigma_i) \in \cC(\cS^{q-1}_{\ids(\sigma_i)})$.
By Lemma \ref{lemmaspheres} and since $\cS^{q-1}_{\ids(\sigma_i)} \subset \cS^q_{0 \hdots q}$, 
there exists a $q$-chain 
$f^{\ids(\sigma')}(\sigma) \in \cC(\cS^q_{0 \hdots q})$ such that 
$\partial f^{\ids(\sigma')}(\sigma) = \sum_{i \in [q] - \ids(\sigma')} (-1)^i f^{\ids(\sigma')}(\sigma_i)$.
We use $f^{\ids(\sigma')}(\sigma)$ to ``symmetrically" define the value of 
$f^s(\sigma'')$ for $dim(\sigma'') = q$, $\vert s \vert = q-2$ and 
$s \subset \ids(\sigma')$. Therefore, 
\begin{equation}
\label{eq2induction}
\sum_{\sigma' \in \sk^{q-3}(\sigma), dim(\sigma')=q-3} \sum_{i \in [q] - \ids(\sigma')} (-1)^i f^{\ids(\sigma')}(\sigma_i) = 
\partial \sum_{\sigma' \in \sk^{q-3}(\sigma), dim(\sigma')=q-3} f^{\ids(\sigma')}(\sigma)
\end{equation}

\paragraph{Case $dim(\sigma') \leq q-4$.} 
By induction hypothesis, for every $i \in [q] - \ids(\sigma')$, 
$$\partial f^{\ids(\sigma')}(\sigma_i) = \sum^{i-1}_{j=0, j \notin \ids(\sigma')} (-1)^j f^{\ids(\sigma')}(\sigma_{ji}) + 
\sum^q_{j=i+1, j \notin \ids(\sigma')} (-1)^{j-1} f^{\ids(\sigma')}(\sigma_{ij})$$
Thus
\begin{eqnarray*}
\partial \sum_{i \in [q] - \ids(\sigma')} (-1)^i f^{\ids(\sigma')}(\sigma_i) &=& 
\sum_{i \in [q] - \ids(\sigma')} \sum^{i-1}_{j=0, j \notin \ids(\sigma')} (-1)^{i+j} f^{\ids(\sigma')}(\sigma_{ji}) \\
& & + \sum_{i \in [q] - \ids(\sigma')} \sum^q_{j=i+1, j \notin \ids(\sigma')} (-1)^{i+j-1} f^{\ids(\sigma')}(\sigma_{ij}) \\
&=& 0
\end{eqnarray*}
Therefore, $\sum_{i \in [q] - \ids(\sigma')} (-1)^i f^{\ids(\sigma')}(\sigma_i)$ is a $(q-1)$-cycle.
By induction hypothesis, $f^{\ids(\sigma')}(\sigma_i) \in \cC(\cS^{q-1}_{\ids(\sigma_i)})$.
By Lemma \ref{lemmaspheres} and since $\cS^{q-1}_{\ids(\sigma_i)} \subset \cS^q_{0 \hdots q}$,
there exists a $q$-chain $f^{\ids(\sigma')}(\sigma) \in \cC(\cS^q_{0 \hdots q})$ such that 
$\partial f^{\ids(\sigma')}(\sigma) = \sum_{i \in [q] - \ids(\sigma')} (-1)^i f^{\ids(\sigma')}(\sigma_i)$.
We use $f^{\ids(\sigma')}(\sigma)$ to ``symmetrically" define the value of 
$f^s(\sigma'')$ for $dim(\sigma'') = q$, $\vert s \vert \leq q-3$ and 
$s \subset \ids(\sigma'')$. Thus, we get 
\begin{equation}
\label{eq3induction}
\sum_{\sigma' \in \sk^{q-4}(\sigma)} \sum_{i \in [q] - \ids(\sigma')} (-1)^i f^{\ids(\sigma')}(\sigma_i) = 
\partial \sum_{\sigma' \in \sk^{q-4}(\sigma)} f^{\ids(\sigma')}(\sigma)
\end{equation}

Combining Equations (\ref{eq2induction}) and (\ref{eq3induction})
\begin{equation}
\label{eq4induction}
\lambda = \partial \sum_{\sigma' \in \sk^{q-3}(\sigma)} f^{\ids(\sigma')}(\sigma)
\end{equation}

Finally, from Equations (\ref{eq0induction}), (\ref{eq5induction}) and (\ref{eq4induction}), we conclude 
$$a(\sigma) - z(\sigma) - d^{\ids(\sigma)}(\partial \sigma) -
\sum_{\sigma' \in \sk^{q-2}(\sigma)} f^{\ids(\sigma')}(\sigma)$$
is a $q$-cycle, hence the lemma holds for faces of dimension $q$.
\end{proof}

\vspace{0.5cm}
\noindent
\emph{Theorem \ref{theo2} (Restated)}
Let $a: \cC(\sigma^n) \to \cC(\cA^n)$ be a non-trivial color-preserving equivariant chain map.
For some set of integers $k_0, \ldots, k_{n-1}$,
\begin{equation*}
a(\partial \sigma^n) \thicksim
\left( 1 + \sum^{n-1}_{q=0} k_q {n+1 \choose q+1} \right) \partial 0^n.
\end{equation*}

\noindent
\begin{proof}
Consider the chain map $z: \cC(\bd(\sigma^n)) \to \cC(\cA^n)$
that maps each simplex $\ang{c_0 \hdots c_i}$ of $\cC(\bd(\sigma^n))$ to
$\ang{(c_0, 0) \ldots (c_i, 0)}$.
Observe that $z(\sigma^n) = \partial 0^n$.
Let $\sigma_i$ denote the oriented face $\ang{0 \hdots \widehat i \hdots n}$ of $\sigma^n$.
Let $S_i$ be the cycle obtained by orienting the $(n-1)$-simplexes of 
$\cS^{n-1}_{0 \hdots \widehat i \hdots n}$ such that its $0$-monochromatic $(n-1)$-simplex is oriented
in increasing $\ids$ order.
By Lemma \ref{lemma1},
$$\alpha_i = a(\sigma_i) - z(\sigma_i) - d^{\ids(\sigma_i)}(\partial \sigma_i) -
\sum_{\sigma \in \sk^{n-3}(\sigma_i)} f^{\ids(\sigma}(\sigma_i)$$
is an $(n-1)$-cycle. Consider the cycle $\alpha_n$. 
By Lemma \ref{lemmahsjacm2}, 
$$\alpha_n \thicksim k_{n-1} S_n$$
for some integer $k_{n-1}$.
It is not hard to see that $\pi^n_i(\sigma_n) = \sigma_i$ and $\pi^n_i(S_n) = S_i$.
Thus, $\pi^n_i(\alpha_n) = \alpha_i$, because $a$, $z$, $d$ and $f$ are equivariant.
Therefore, 
$$\pi^n_i(\alpha_n) = \alpha_i \thicksim k_{n-1} \pi^n_i(S_n) = k_{n-1} S_i$$
and by Lemma \ref{lemmahsjacm3}
$$\alpha_i \thicksim (-1)^i k_{n-1} S_i$$

Considering the alternating sign sum over all
$(n-1)$-faces of $\sigma^n$, we get
$$\sum^n_{i=0} (-1)^i \left( a(\sigma_i) - z(\sigma_i) - d^{\ids(\sigma_i)}(\partial \sigma_i) -
\sum_{\sigma \in \sk^{n-3}(\sigma_i)} f^{\ids(\sigma)}(\sigma_i) \right)
\thicksim \sum^n_{i=0} (-1)^i (-1)^i k_{n-1} \partial 0^n$$
hence
$$a(\partial \sigma^n) - z(\partial \sigma^n) - \sum^{n}_{i=0} (-1)^i d^{\ids(\sigma_i)}(\partial \sigma_i) - 
\sum^{n}_{i=0} (-1)^i \sum_{\sigma \in \sk^{n-3}(\sigma_i)} f^{\ids(\sigma)}(\sigma_i) \thicksim k_{n-1} (n+1) \partial 0^n$$
And since $z(\partial \sigma^n) = \partial 0^n$
$$a(\partial \sigma^n) \thicksim \left(1 + k_{n-1} (n+1) \right) \partial 0^n + 
\sum^{n}_{i=0} (-1)^i d^{\ids(\sigma_i)}(\partial \sigma_i) + 
\sum^{n}_{i=0} (-1)^i \sum_{\sigma \in \sk^{n-3}(\sigma_i)} f^{\ids(\sigma)}(\sigma_i)$$
Notice that if we prove
\begin{equation}
\label{eq1}
\sum^{n}_{i=0} (-1)^i d^{\ids(\sigma_i)}(\partial \sigma_i) \thicksim k_{n-2} {n+1 \choose n-1} \partial 0^n
\end{equation} 
and 
\begin{equation}
\label{eq2}
\sum^{n}_{i=0} (-1)^i \sum_{\sigma \in \sk^{n-3}(\sigma_i)} f^{\ids(\sigma)}(\sigma_i) 
\thicksim \sum^{n-3}_{q=0} k_q {n+1 \choose q+1} \partial 0^n
\end{equation}
then 
$$a(\partial \sigma^n) \thicksim \left(1 + \sum^{n-1}_{q=0} k_q {n+1 \choose q+1} \right) \partial 0^n$$

\paragraph{Proof of equation (\ref{eq1}).}
For $i, j \in [n]$ such that $i < j$, let $\alpha_{ij}$ be 
$(-1)^{i+j} ( d^{\ids(\sigma_j)}(\sigma_{i j}) - d^{\ids(\sigma_i)}(\sigma_{i j}) )$,
where $\sigma_{i j}$ is the $(n-2)$-face $\ang{0 \hdots \widehat i \hdots \widehat j \hdots n}$
of $\sigma^n$. The proof of Lemma \ref{lemma1} shows that 
$$\sum^{n}_{i=0} (-1)^i d^{\ids(\sigma_i)}(\partial \sigma_i) = \sum^n_{i=0} \sum^n_{j+1} \alpha_{ij}$$
and $\alpha_{ij}$ is an $(n-1)$-cycle.

Consider $i, j \in [n]$ such that $i < j < n$. We have that 
\begin{eqnarray*}
\alpha_{ij} &=& (-1)^{i+j} ( d^{\ids(\sigma_j)}(\sigma_{i j}) - d^{\ids(\sigma_i)}(\sigma_{i j}) ) \\
\alpha_{i j+1} &=& (-1)^{i+j+1} ( d^{\ids(\sigma_{j+1})}(\sigma_{i j+1}) - d^{\ids(\sigma_i)}(\sigma_{i j+1}) )
\end{eqnarray*}
It is easy to see that $\pi^{j+1}_j(\sigma_{i j}) = \sigma_{i j+1}$,
$\pi^{j+1}_j(\sigma_j) = \sigma_{j+1}$ and $\pi^{j+1}_j(\sigma_i) = - \sigma_i$.
Thus, $\pi^{j+1}_j(\alpha_{i j}) = \alpha_{i j+1}$, because $d$ is equivariant.
By Lemma \ref{lemmahsjacm2}, for some integer $k_{ij}$,
\begin{equation}
\label{eqx}
\alpha_{ij} \thicksim (-1)^i k_{ij} S_i
\end{equation}
It can be easily proved that $\pi^{j+1}_j(S_i) = -S_i$. 
Applying $\pi^{j+1}_j$ on both sides of Equation (\ref{eqx}) and then multiplying by $-1$,
we get 
\begin{equation}
\label{eqy}
\alpha_{ij+1} \thicksim (-1)^i k_{ij+1} S_i
\end{equation}
By Lemma \ref{lemmahsjacm3} and Equations (\ref{eqx}) and (\ref{eqy}),
$\alpha_{ij} \thicksim k_{ij} \partial 0^n$ and $\alpha_{i j+1} \thicksim k_{ij} \partial 0^n$.
A similar analysis gives that, for every $i, j \in [n]$ such that $i < j-1$,
$\alpha_{ij} \thicksim k_{ij} \partial 0^n$ and $\alpha_{i+1 j} \thicksim k_{ij} \partial 0^n$.

We can repeatedly use these two arguments to prove that  
$\alpha_{ij} \thicksim k_{ij} \partial 0^n$ and $\alpha_{i' j'} \thicksim k_{ij} \partial 0^n$,
for every $i, i', j, j' \in [n]$, $i < j$ and $i' < j'$. Therefore, 
$$\sum^n_{i=0} \sum^n_{j+1} \alpha_{ij} \thicksim {n+1 \choose n-1} k_{n-2} \partial 0^n$$
for some integer $k_{n-2}$.

\paragraph{Proof of equation (\ref{eq2}).}
The argument is very similar to the one used for Equation (\ref{eq1}).
The proof of Lemma \ref{lemma1} shows that 
$$\sum^{n}_{i=0} (-1)^i \sum_{\sigma \in \sk^{n-3}(\sigma_i)} f^{\ids(\sigma)}(\sigma_i) =
\sum_{\sigma \in \sk^{n-3}(\sigma^n)} \sum_{i \in [n] - \ids(\sigma)} (-1)^i f^{\ids(\sigma)}(\sigma_i)$$ 
%= \sum_{s \in {[n] \choose \leq n-2}} \sum_{i \in [n] - s} (-1)^i f^s(\sigma_i)$$
Also it shows that $\sum_{i \in [n] - \ids(\sigma)} (-1)^i f^{\ids(\sigma)}(\sigma_i)$ is an $(n-1)$-cycle.
For each $\sigma \in \sk^{n-3}(\sigma^n)$, 
let $\alpha_\sigma$ be the cycle $\sum_{i \in [n] - \ids(\sigma)} (-1)^i f^{\ids(\sigma)}(\sigma_i)$.

Consider $\sigma, \sigma' \in \sk^{n-3}(\sigma^n)$ of same dimension
such that for some $P \subset [n]$ and $j \in [n]$, $\ids(\sigma) = P \cup \{ j \}$, $\ids(\sigma') = P \cup \{ j+1 \}$
and $j, j+1 \notin P$. Note
\begin{eqnarray*}
\alpha_\sigma &=& \sum_{i \in [n] - \ids(\sigma)} (-1)^i f^{\ids(\sigma)}(\sigma_i) = (-1)^{j+1} f^{\ids(\sigma)}(\sigma_{j+1}) + \sum_{i \in [n] - P} (-1)^i f^{\ids(\sigma)}(\sigma_i)\\
\alpha_{\sigma'} &=& \sum_{i \in [n] - \ids(\sigma')} (-1)^i f^{\ids(\sigma')}(\sigma_i) = (-1)^j f^{\ids(\sigma')}(\sigma_j) + \sum_{i \in [n] - P} (-1)^i f^{\ids(\sigma')}(\sigma_i)
\end{eqnarray*}
It is easy to see that $\pi^{j+1}_j(\sigma) = \sigma'$, $\pi^{j+1}_j(\sigma_{j+1}) = \sigma_j$ and 
$\pi^{j+1}_j(\sigma_i) = - \sigma_i$ for each $i \in [n] - P$. 
Then, $\pi^{j+1}_j(\alpha_\sigma) = - \alpha_{\sigma'}$, since $f$ is equivariant.

Fix an $i \in [n] - \ids(\sigma)$. By Lemma \ref{lemmahsjacm2}, for some integer $k_\sigma$
\begin{equation}
\label{eqa}
\alpha_\sigma \thicksim (-1)^i k_\sigma S_i
\end{equation}
It can be easily proved that $\pi^{j+1}_j(S_i) = - S_i$.
Applying $\pi^{j+1}_j$ on both sides of Equation (\ref{eqb}) and then multiplying by $-1$, we get
\begin{equation}
\label{eqb}
\alpha_{\sigma'} \thicksim (-1)^i k_\sigma S_i
\end{equation}
By Lemma \ref{lemmahsjacm3} and Equations (\ref{eqa}) and (\ref{eqb}), 
$\alpha_\sigma \thicksim k_\sigma \partial 0^n$ and $\alpha_{\sigma'} \thicksim k_\sigma \partial 0^n$.
We can repeatedly use this argument to prove that 
$\alpha_\sigma \thicksim k_\sigma \partial 0^n$ and $\alpha_{\sigma'} \thicksim k_\sigma \partial 0^n$,
for every $\sigma, \sigma' \in \sk^{n-3}(\sigma^n)$ of same dimension.
Therefore,
$$\sum_{\sigma \in \sk^{n-3}(\sigma^n)} \sum_{i \in [n] - \ids(\sigma)} (-1)^i f^{\ids(\sigma)}(\sigma_i)
\thicksim \sum^{n-3}_{q=0} k_q {n+1 \choose q+1} \partial 0^n$$
\end{proof}

\subsection{Proofs of Section \ref{sec:sufficiency}}

\begin{lemma}
Let $a: \cC(\sigma^n) \rightarrow \cC(\cA^n)$ be the chain map induced
by a chromatic and binary colored subdivision $\chi(\sigma^n)$ without monochromatic $n$-simplexes
produced by the construction in \cite{crpodc08}. Then, $a$ is non-trivial, color-preserving
and equivariant.
\end{lemma}
\begin{proof}
Let $\cO^n$ be the complex $\cA^n$ with the monochromatic simplexes
$\set{(0,0), \hdots, (n,0)}$ and $\{ (0,1), \hdots, \\ (n,1) \}$.
Observe that $a$ is also a chain map $\cC(\sigma^n) \rightarrow \cC(\cO^n)$.
For technical reasons, we think of $a$ in this way.

First, since the subdivision $\chi(\sigma^n)$ is chromatic, clearly $a$ is non-trivial and color-preserving.
By induction on $q$, we prove the following proposition.

\begin{proposition}
\label{prop1}
The restriction $a \vert_{\cC(\sk^q(\sigma^n))}$, $0 \leq q \leq n$, is equivariant.
\end{proposition}

By symmetry of the binary coloring of $\chi(\sigma^n)$,
Proposition \ref{prop1} clearly holds for $q=0$. Suppose that Proposition \ref{prop1} holds for dimension $q-1$.
We prove it holds for dimension $q$.

By symmetry of the binary coloring of $\chi(\sigma^n)$,
for the face $\sigma = \ang{0 \hdots q}$ of $\sigma^n$
we have that $a \circ \pi^q_i(\sigma_q) = a(\sigma_i) = \pi^q_i \circ a(\sigma_q)$,
where $0 \leq i \leq q$ and $\sigma_i = \ang{0 \hdots \widehat i \hdots q}$. 
Therefore, if we prove that $\pi \circ a(\sigma_q) = a \circ \pi(\sigma_q)$ for every $\pi \in \cS_n$,
then $\pi \circ a(\sigma_i) = a \circ \pi(\sigma_i)$, since $a(\sigma_i) = \pi^q_i \circ a(\sigma_q)$.

Consider the face $\sigma = \ang{0 \hdots q}$ of $\sigma^n$.
Let $L_q$ be $\{ \tau \vert \tau \in \sk^q(\cO^n) \hbox{ and } \ids(\tau) = [q] \}$.
For $\tau \in L_q$, let $\#1(\tau)$ be the number of its vertexes with binary color $1$,
and let $\inv(\tau, i)$, $0 \leq i \leq q$, denote the simplex of $L_q$ with the same
vertexes as $\tau$ but with the vertex with $\id$ $i$ having the opposite binary coloring
to the binary coloring of the vertex with $\id$ $i$ of $\tau$.
For $0 \leq k \leq q+1$, let $L_{q,k}$ denote the set $\{ \tau \vert \tau \in L_q \hbox{ and } \#1(\tau) = k \}$. 
Thus $\vert L_{q,k} \vert = {q+1 \choose k}$.
Since $a$ is color-preserving, we can write 
$$a(\sigma) = \sum_{\tau \in L_q} k_{\tau} \tau$$
where $k_{\tau} \in \mathbb{Z}$. Obviously if $q = n$ then
$k_{ \{ (0,0), \hdots, (n,0) \} } = k_{ \{ (0,1), \hdots, (n,1) \} } = 0$,
since $\cA^n$ does not have monochromatic $n$-simplexes.
We prove the following proposition.

\begin{proposition}
\label{prop2}
For every $\tau, \tau' \in L_{q, k}$, $k_\tau = k_{\tau'}$, $0 \leq k \leq q+1$. 
\end{proposition}

For example, for $\sigma = \ang{012}$ and $k=2$, Proposition \ref{prop2} says that if 
$\ang{(0,0) (1,1) (2,1)}$ appears in $a(\sigma)$ with coefficient $\ell$,
then $\ang{(0,1) (1,1) (2,0)}$ and $\ang{(0,1) (1,0) (2,1)}$ appear in $a(\sigma)$
with coefficient $\ell$ too. It is not hard to see that this proves 
$a \circ \pi(\sigma) = \pi \circ a(\sigma)$ for every $\pi \in \cS_n$, 
hence Proposition \ref{prop1} holds for $q$.

We proceed by induction on $k$. For $k=0$ we have that $\vert L_{q,k} \vert = 1$,
thus Proposition \ref{prop2} trivially holds. Suppose Proposition \ref{prop2} holds for $k-1$. 
We prove it holds for $k$. 

Notice that
$$\partial a(\sigma) = \sum_{\tau \in L_q} k_{\tau} \partial \tau = 
\sum_{\tau \in L_q} k_{\tau} \sum ^q_{i=0} (-1)^i \tau_i = 
a(\partial \sigma) = \sum^q_{i=0} (-1)^i a(\sigma_i)$$
where $\tau_i = \langle (0,b_0) \hdots \widehat{(i,b_i)} \hdots (q,b_q) \rangle$
for $\tau = \langle (0,b_0) \hdots (i,b_i) \hdots (q,b_q) \rangle$.
Consider a simplex $\tau \in L_q$ and $i \in \{ 0, \hdots, q\}$.
Observe that the $(q-1)$-simplex $\tau_i$ appears in $\partial a(\sigma)$
with coefficient $(-1)^i (k_\tau + k_{\inv(\tau, i)})$,
since $\tau_i$ is face of $\tau$ and $\inv(\tau, i)$.
Moreover, $\tau_i$ appears in $a(\sigma_i)$ with
coefficient $k_\tau + k_{\inv(\tau, i)}$,
because $\partial a(\sigma) = a(\partial \sigma)$ and $a$ is color-preserving.
Also notice that either $\#1(\tau) = \#1(\tau_i)$ and $\#1(\inv(\tau, i)) = \#1(\tau_i) + 1$,
or $\#1(\tau) = \#1(\tau_i) + 1$  and $\#1(\inv(\tau, i)) = \#1(\tau_i)$.

Consider the set $N = \{ \tau \vert \tau \in L_{q, k} \hbox{ and } \#1(\tau_q) = k-1 \}$.
Note $\vert N \vert = {q \choose k-1}$. 
For each $\tau \in N$,
observe that $\#1(\inv(\tau, q)) = k-1$, hence $\inv(\tau, q) \in L_{q,k-1}$.
Consider a simplex $\tau \in N$.
As noticed above, $\tau_q$ appears in $a(\sigma_q)$
with coefficient $k_{\tau} + k_{\inv(\tau, q)}$. 
Consider $i \in \{ 0, \hdots, q \}$.
Let $\rho_i$ and $\rho$ be the simplexes
$\pi^q_i(\tau_q)$ and $\pi^q_i(\tau)$.
Observe that $\rho_i$ is a face of $\rho$,
$\#1(\rho_i) = k-1$ and $\#1(\rho) = k$.
As for $\tau_q$, we have that $\rho_i$ appears in 
$a(\sigma_i)$ with coefficient $k_{\rho} + k_{\inv(\rho, i)}$,
where $\sigma_i = \pi^q_i(\sigma_q)$.
By the induction hypothesis,
$a \vert_{\cC(\sk^{q-1}(\sigma^n))}$ is equivariant,
hence $a \circ \pi^q_i(\sigma_q) = a(\sigma_i) = \pi^q_i \circ a(\sigma_q)$.
Therefore,  $k_{\tau} + k_{\inv(\tau, q)} = k_{\rho} + k_{\inv(\rho, i)}$.
Moreover, $k_{\inv(\tau, q)} = k_{\inv(\rho, i)}$
because $\#1(\inv(\tau, q)) = \#1(\inv(\rho, i)) = k-1$ and, by the induction hypothesis,
Proposition \ref{prop2} holds for $k-1$.
Thus, we get $k_\tau = k_\rho$.

For each $\tau \in N$, let $M_\tau$ be $\{ \pi^q_i(\tau) \vert 0 \leq i \leq q \}$.
The previous paragraph proved that for every $\rho, \rho' \in M_\tau$, $k_\rho = k_{\rho'}$.
It is not hard to see that $\vert M_{\tau} \vert = (q+1) - (k-1)$ for every $\tau \in N$, 
and $L_{q,k} = \cup_{\tau \in N} M_\tau$.
Moreover, we have that the sets $M_\tau$'s are not a partition of $L_{q,k}$ because  
$$\frac{ {q+1 \choose k} }{ ((q+1)-(k-1)) {q \choose k-1} } = \frac{ q+1 }{ k((q+1)-(k-1)) } < 1$$
Thus, these sets intersect each other, hence $\tau, \tau' \in L_{q, k}$, $k_\tau = k_{\tau'}$.
This completes the proof.
\end{proof}

\end{document}